\documentclass[10pt,conference]{IEEEtran}
\usepackage[top=1in,bottom=0.75in,right=0.75in,left=0.75in]{geometry}
\usepackage{geometry}
\usepackage{graphicx}
\usepackage{amssymb, amsmath, graphicx, cite, color, epsfig, subfigure}

\newtheorem{theorem}{Theorem}
\newtheorem{lemma}[theorem]{Lemma}

\newtheorem{remark}[theorem]{Remark}

\begin{document}
\title{On Constant Gaps for the Two-way Gaussian Interference Channel}
\author{\IEEEauthorblockN{Zhiyu Cheng, Natasha Devroye\\
University of Illinois at Chicago \\
zcheng3, devroye@uic.edu
\thanks{The work of Z. Cheng and N. Devroye was partially supported by NSF under award 1053933. The contents of this article are solely the responsibility of the authors and do not necessarily represent the official views of the NSF. 
}
}}

\maketitle


\begin{abstract}
We introduce the two-way Gaussian interference channel in which there are four nodes with four independent messages: two-messages to be transmitted over a Gaussian interference channel in the $\rightarrow$ direction, simultaneously with two-messages to be transmitted over an interference channel (in-band, full-duplex) in the $\leftarrow$ direction. In such a two-way network, all nodes are transmitters and receivers of messages, allowing them to adapt current channel inputs to previously received channel outputs. We propose two new outer bounds on the symmetric sum-rate for the two-way Gaussian interference channel with complex channel gains: one under full adaptation (all 4 nodes are permitted to adapt inputs to previous outputs), and one under partial adaptation (only 2 nodes are permitted to adapt, the other 2 are restricted). We show that simple non-adaptive schemes such as the Han and Kobayashi scheme, where inputs are functions of messages only and not past outputs, utilized in each direction are sufficient to achieve within a constant gap of these fully or partially adaptive outer bounds for all channel regimes. 
\end{abstract}

\section{Introduction}

In two-way networks, multiple pairs of possibly interfering users wish to exchange pairs of messages. While this is a natural form of communication in wireless networks, from an information theoretic perspective such two-way networks are challenging to deal with and as such, most two-way exchanges are treated as two one-way exchanges. 
What makes such two-way communications challenging are the possibilities that stem from having  nodes act as both  sources and destinations of messages. This permits them to adapt their channel inputs to their past received signals. Such two-way adaptation or interaction was first considered in the point-to-point two-way channel by Shannon \cite{Shannon:1961}, but capacity remains unknown in general. 

However, encouragingly, the capacity regions of several specific point-to-point two way channel models is known. What is common to these models is  that the interaction between ones own signal and that of the other user may be resolved.  For example, in the two-way modulo 2 binary adder channel where channel outputs  $Y_1 = Y_2 = X_1 \oplus  X_2$ for binary inputs $X_1,X_2$ and  $\oplus$ modulo 2 addition, the capacity region is one bit per user per channel use, as each user is able to ``undo'' the effect of the other, something that is not possible (at least not in one channel use) for the  elusive binary multiplier channel with $Y_1=Y_2 = X_1 X_2$. In the binary modulo 2 adder channel, information independently flows in the $\rightarrow$ and the $\leftarrow$ ``directions''  and nodes need not interact, or adapt their current inputs to past outputs,  to achieve capacity.  In a similar fashion, the capacity of a two-way Gaussian point-to-point channel is equal to two parallel Gaussian channels, which may be achieved without the use of adaptation at the nodes \cite{Han:1984}. In general then, one may ask whether there exist two-way {\it networks} rather than point-to-point channels where capacity may be obtained in a similar fashion, and where adaptation does not increase the capacity region. 
\subsection{Previous work on two-way deterministic networks.} In our previous work, we have demonstrated several  examples of  multi-user two-way channels  where, even though nodes may adapt current inputs to past outputs, this is {\it not} beneficial from a capacity region perspective. 
In \cite{zcheng_allerton, zcheng_ISIT} we considered three multi-user two-way channel models:
\begin{itemize}
\item  the {\bf two-way Multiple Access / Broadcast  channel (MAC/BC)} in which there are 4 messages and 3 terminals forming a MAC channel in the $\rightarrow$ direction (2 messages) and a BC channel in the opposite $\leftarrow$ direction (2 messages); 
\item  the {\bf two-way Z channel} in which there are 6 messages and  4 terminals forming a Z channel in the $\rightarrow$ direction (3 messages) and another Z channel in the opposite $\leftarrow$ direction (3 messages); 
\item the {\bf two-way interference channel (IC)} with 4 messages and 4 terminals forming an IC in the $\rightarrow$ direction (2 messages) and another IC in the $\leftarrow$ direction (2 messages).
\end{itemize}
In particular, in \cite{zcheng_allerton} we obtained the capacity regions of the deterministic, binary modulo 2 adder models for all three channels, where it was shown that adaptation at the nodes does not increase the capacity regions beyond non-adaptive schemes. In follow-up work in \cite{zcheng_ISIT} we considered a slightly more general class of deterministic channels: the linear deterministic channels in the spirit of \cite{Avestimehr:2007:ITW}. There, we showed that again, for the two-way MAC/BC and two-way Z channels that adaptation does not increase the capacity region and the capacity region is that of two one-way channel models operating in parallel. For the two-way linear deterministic interference channel, we showed that if we allow only 2 of the four nodes to adapt (which we termed ``partial adaptation''), then the capacity region is the same as if none of the 4 nodes were able to adapt, i.e. partial adaptation is useless from a capacity perspective.

\begin{figure}
\centerline{\includegraphics[width=10cm]{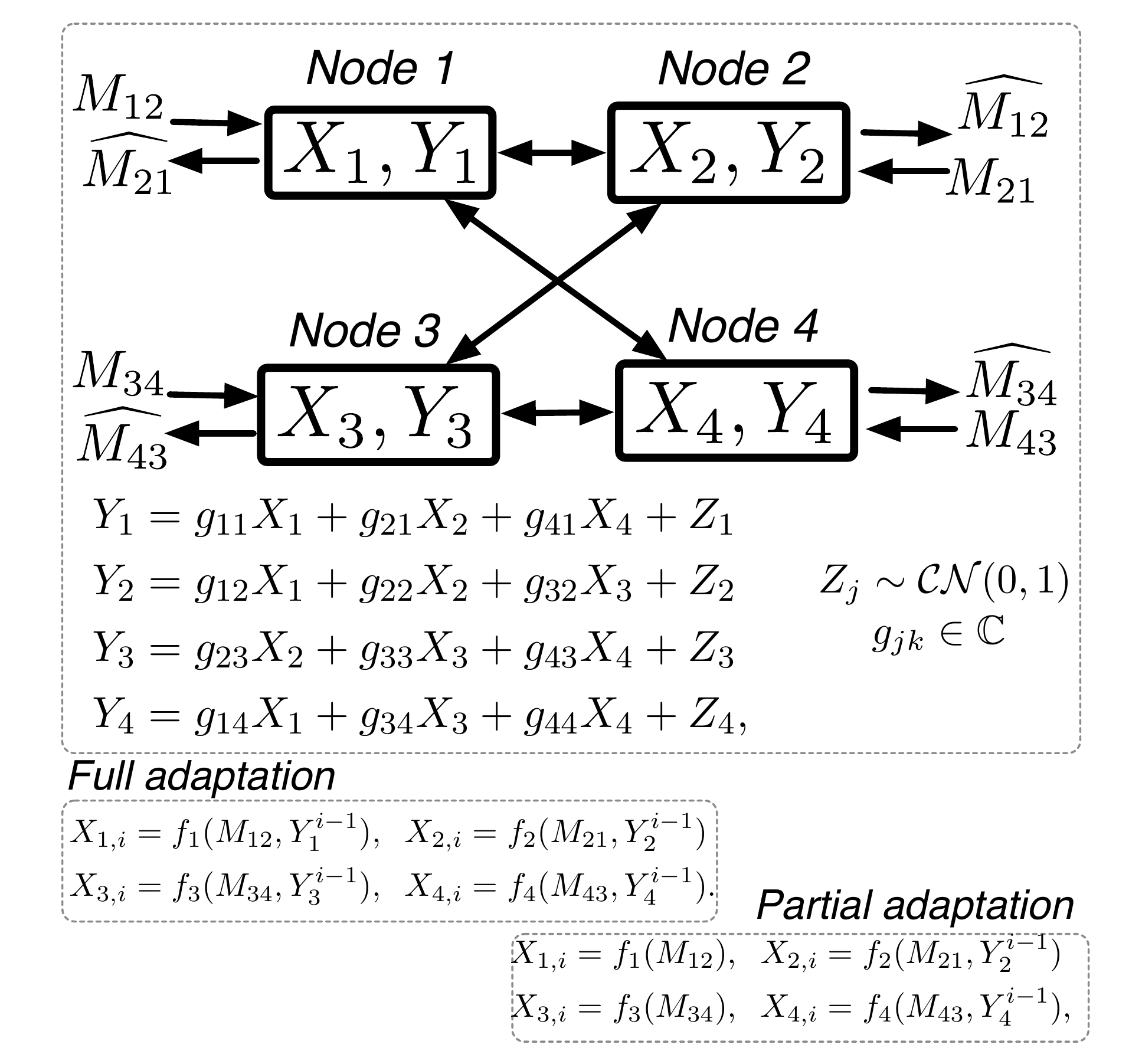}}
\vspace{-0.5cm}
\caption{The two-way Gaussian interference channel under full and partial adaptation constraints.}
\label{fig:channel}
\end{figure}

\subsection{Contributions}
In all our previous work, we considered {\it deterministic} channel models. We now consider a new (not considered before) noisy channel model: the two-way Gaussian interference channel. As a first step, we consider 
the {\bf symmetric two-way Gaussian IC} where all ``direct'' links are equal and all ``cross-over'' links are equal. We derive new, computable outer bounds for the symmetric sum-rates for this Gaussian channel model and show that: a) adaptation is useless in very strong interference for the partially adaptive model, b) in strong but not very strong interference, non-adaptive schemes perform to within 1 bit per user per direction of the fully adaptive capacity region, and c) the particular non-adaptive Han and Kobayashi scheme of \cite{etkin_tse_wang} employed in each direction, achieves to within a  constant gap (2 bits per user per direction maximally) of fully or partially adaptive outer bounds in all other regimes. In general, when all nodes are permitted to adapt, we do not believe that a non-adaptive scheme will achieve to within a constant gap for all regimes but this is left open.  
Our emphasis, as with our prior work \cite{zcheng_allerton, zcheng_ISIT} is on demonstrating when {\it adaptive} schemes do not increase capacity, and when, even if adaptation is permitted, it does not significantly increase the capacity region. 

\subsection{Related Work} 


We focus only on work related to the two-way interference channel rather than two-way channels in general; a more extensive list of references related to two-way networks may be found in \cite{zcheng_allerton, zcheng_ISIT,  zcheng_arxiv}. 

The two-way Gaussian interference channel is naturally related to one-way interference channels with/without feedback. 
The capacity region of the one-way modulo 2 adder  IC is known  \cite{ElGamalKim:book} and is a special example of a more general class of deterministic IC for which capacity is known \cite{elgamal_det_IC}, including the one-way linear deterministic IC \cite{bresler_tse}.  
The work here is also related to one-way ICs with perfect output feedback \cite{sahai2009channel, Changho2010}, with rate-limited feedback \cite{Vahid-IC-rate-FB},  and interfering feedback \cite{sahai2009channel, Suh2012}\footnote{We will refer to the 4 message two-way IC considered here as the ``two-way IC'' and the 2 message channel of \cite{Suh2012} -- considered in all sections but Section VI --  as the ``two-way interference channel with interfering feedback'' to emphasize that the rates are flowing in one direction.}. In all these channel models only two messages are present and the ``feedback'' links, whether perfect, noisy, or interfering still serve only to further rates in the forward direction. The tradeoff between sending new information versus feedback on each of the links is not addressed.  The only other example of such a 4-message two-way interference channel besides our prior work \cite{zcheng_allerton, zcheng_CISS, zcheng_ISIT} is in Section  VI of \cite{Suh2012}, where an example of a linear deterministic scheme in a specific regime is  provided which shows that, at least for one particular asymmetric linear deterministic two-way IC  with weak  interference in the $\rightarrow$ and strong interference in the $\leftarrow$ direction, that adaptation can significantly improve the capacity region over non-interaction. The general capacity region of the linear deterministic two-way IC (with 4 messages)  remains open in general despite the example in \cite{Suh2012} and the results in \cite{zcheng_ISIT}. 
This is the first work to consider the two-way Gaussian interference channel.

%

\section{Channel Model}
\label{model}

A graphical depiction of the two-way Gaussian interference channel is provided in Fig. \ref{fig:channel}. 
There are 4 nodes: transmitters 1 and 3 send messages $M_{12}$ and $M_{34}$ to receivers 2 and 4, respectively, forming an IC in the $\rightarrow$ direction. Similarly,  transmitters 2 and 4 send messages $M_{21}$ and $M_{43}$ to receivers 1 and 3 respectively, forming another IC in the $\leftarrow$ direction. All messages $M_{jk}$ from node $j$ to node $k$ are independent and uniformly distributed over $\mathcal{M}_{jk} : =\{1,2,\cdots 2^{nR_{jk}}\}$ (for appropriate $j,k$) and $R_{jk}$ is the rate of transmission from node $j$ to node $k$. 

 All channels are assumed to be memoryless and at each channel use, are described by 
 \begin{align*}
&Y_1=g_{11}X_1+g_{21}X_2+g_{41}X_4+Z_1\\
&Y_2=g_{12}X_1+g_{22}X_2+g_{32}X_3+Z_2\\
&Y_3=g_{23}X_2+g_{33}X_3+g_{43}X_4+Z_3\\
&Y_4=g_{14}X_1+g_{34}X_3+g_{44}X_4+Z_4,
\end{align*}
where $g_{jk}$, for $j,k\in \{1,2,3,4\}$ are the complex channel gains. 
 Let $X_j$ and $Y_k$ denote the channel input of node $j$ and output at node $k$ used to describe the model (per channel use). Let $X_{j,i} \; (Y_{j,i})$ denote the channel input (output) at node $j$ at channel use $i$, and $X_j^n : = (X_{1,1}, X_{1,2}, \cdots X_{1,n})$. 
We assume the power constraints $E[|X_j|^2]\leq P_j=1, j\in\{1,2,3,4\}$, and independent, identically distributed complex Gaussian noise $Z_j\sim \mathcal{CN}(0, 1)$ at all nodes  $j\in (1,2,3,4)$, which may be done without loss of generality. 

We say that the two-way Gaussian interference channel operates under ``full adaptation'' if we allow
\begin{align}
& X_{1,i}=f_{1,i}(M_{12},Y_1^{i-1}), \;\; X_{2,i}=f_{2,i}(M_{21},Y_2^{i-1})\\
& X_{3,i}=f_{3,i}(M_{34},Y_3^{i-1}), \;\; X_{4,i}=f_{4,i}(M_{43},Y_4^{i-1}),
\end{align}
for $f_{j,i}$ deterministic encoding functions for $1\leq i\leq n$ ($n$ is the blocklength). 
Similarly, it operates under ``partial adaptation'' if we only allow the following:  
\begin{align}
& X_{1,i}=f_{1,i}(M_{12})\label{eq:p1}, \;\; X_{2,i}=f_{2,i}(M_{21},Y_2^{i-1})\\
& X_{3,i}=f_{3,i}(M_{34}), \;\; X_{4,i}=f_{4,i}(M_{43},Y_4^{i-1}),\label{eq:p2}
\end{align}
i.e.  nodes 1 and 3 are ``restricted'' \cite{Shannon:1961}.  By symmetry, we may alternatively allow nodes $2$ and $4$ to be restricted and $1,3$ to be fully adaptive; whether allowing $1,2$ or $1,4$ to be restricted and the complement fully adaptive remains an open problem.
Receiver $k$ uses a decoding function $\mathcal{Y}_k^n \times \mathcal{M}_{ki} \rightarrow\mathcal{\widehat{M}}_{jk}$ to obtain an estimate $\widehat{M}_{jk}$ of the transmitted message $M_{jk}$ given knowledge of its own message(s) $M_{ki}$ for suitable $i$ (based on Fig. \ref{fig:channel}).  The capacity region is the supremum over all rate tuples for which there exist encoding and decoding functions (of the appropriate rates) which simultaneously drive the probability that any of the estimated messages is not equal to the true message, to zero as $n\rightarrow \infty$. 




Furthermore, we define ${\tt SNR}_{12}=|g_{12}|^2, {\tt SNR}_{21}=|g_{21}|^2, {\tt SNR}_{34}=|g_{34}|^2, {\tt SNR}_{43}=|g_{43}|^2$, and ${\tt INR}_{14}=|g_{14}|^2, {\tt INR}_{41}=|g_{41}|^2, {\tt INR}_{23}=|g_{23}|^2, {\tt INR}_{32}=|g_{32}|^2$. Note that we have kept the ``self-interference'' terms such as $g_{11}X_1$ in the expression of $Y_1$ (for example). In this Gaussian model, it is clear that since node $1$ knows $X_1$ we may equivalently remove this self-interference term due to the additive nature of the channel and hence including it is unnecessary. However,  we leave it in our expressions to emphasize the fact that we can cancel or subtract out a node's ``self-interference'' in all converses.  We speculate that this is one of the reasons two-way channels of this form, as seen in the Gaussian two-way channel as well \cite{Han:1984}, are easier to deal with. 

{\it Symmetric capacity.} We are interested in the symmetric capacity when all the SNRs equal a given ${\tt SNR}$, and all the INRs equal a given ${\tt INR}$. For full adaptation, due to the symmetry,  we consider the per-user rates $R_{sym} = \frac{R_{12}+R_{34}}{2} = \frac{R_{21}+ R_{43}}{2}$.  Under partial adaptation, there is only partial symmetry (nodes 1 and 3 are fixed, while 2 and 4 are not). Hence, we will consider the per user rates $R_{sym\rightarrow} = \frac{R_{12}+R_{34}}{2}$ and $R_{sym\leftarrow} = \frac{R_{21}+R_{43}}{2}$ for the $\rightarrow$ and $\leftarrow$ directions respectively. 

\section{Outer bounds}
We now present two outer bounds for the two-way Gaussian IC under full and partial adaptation respectively. These bounds are either within a constant gap, or sufficient to show the capacity depending on different regimes. We will derive general outer bounds, imposing symmetry only in the final step.

We note that while the converses and the steps are new and exploit carefully chosen genies, when we evaluate these by further outer-bounding our outer-bounds, interestingly, we  sometimes re-obtain some of the outer bounds of the interference channel \cite{etkin_tse_wang} {\it or} the interference channel with feedback \cite{Changho2010}. This  in turn is sufficient to achieve capacity to within a constant gap (which we emphasize, sometimes is limited to {\it partial} adaptation for some of the weak interference regimes but this will be explicitly mentioned when it is the case). 

We first prove a Lemma relevant in partial adaptation which is central to many of our converses.
\begin{lemma}
Under partial adaptation \eqref{eq:p1} -- \eqref{eq:p2}, for some deterministic functions $f_5$ and $f_6$, 
\begin{align}
X_{2,i}& = f_5(M_{12}, M_{21}, M_{34}, Z_2^{i-1}) \perp M_{43}, \;\; \forall i \label{eq:X2}\\
X_{4,i}& = f_6(M_{43}, M_{34}, M_{12},Z_4^{i-1}) \perp M_{21}, \;\; \forall i \label{eq:X4}
\end{align}
where $\perp$ denotes independence.
\label{lemma:partialG}
\end{lemma}
\begin{proof}
Note that $X_{2,i} = f_2(M_{21}, Y_2^{i-1})$ and $Y_2^{i-1} = g_{12}X_{1}^{i-1} + g_{22}X_2^{i-1}+g_{32}X_3^{i-1}+Z_2^{i-1}$. Since $X_1^{i-1}$ and $X_3^{i-1}$ are functions only of $M_{12}$ and $M_{34}$ respectively, we may conclude that  there exists a function $f^*$ such that $X_{2,i} = f^*(M_{21}, M_{12}, M_{34}, X_2^{i-1},Z_2^{i-1})$. Iterating this argument, and noting that $X_{2,1}$ is only a function of $M_{21}$, we obtain the lemma. The result for $X_{4,i}$ follows similarly. That $X_{2,i}$ is independent of $M_{43}$ follows since $M_{43}$ is independent of all the arguments inside $f^*$. 
\end{proof}

\medskip

\begin{theorem} {\it Outer bound: full adaptation.}
For the two-way Gaussian symmetric IC under full adaptation, any achievable symmetric rate $R_{sym} = \frac{R_{12}+R_{34}}{2} = \frac{R_{21}+R_{43}}{2}$, achievable by each user,  satisfies, 
\begin{align}
R_{sym}
&\leq \frac{1}{2} \log \left(1+{\tt SNR}+{\tt INR}+2\sqrt{{\tt SNR}\times {\tt INR}}\right)\nonumber \\
& \;\;\;\;\; +\frac{1}{2}\log \left(1+\frac{{\tt SNR}}{1+{\tt INR}}\right) \label{R_strong_E}
\end{align}
\label{thm:outer-full}
\end{theorem}
\begin{proof}
It is sufficient to consider $R_{12}+R_{34}$ due to symmetry. This bound is inspired by the corresponding sum-rate bound in the linear deterministic model \cite{zcheng_ISIT}, i.e., we add asymmetric genie $Y_2^n$ at node 4. Notice the genie $Z_1^n$ in the conditioning of both terms as well. 
{\small \begin{align}
&n(R_{12}+R_{34}-\epsilon) \leq I(M_{12};Y_2^n|M_{21},M_{43},Z_1^n)\notag \\
&+I(M_{34};Y_4^n,Y_2^n|M_{12},M_{21},M_{43},Z_1^n)\notag\\
&\overset{(a)}{=}I(M_{12};Y_2^n|M_{21},M_{43},Z_1^n)+I(M_{34};Y_2^n|M_{21},M_{12},M_{43},Z_1^n)\notag\\
&+\sum_{i=1}^n[H(g_{34}X_{3,i}+Z_{4,i}|M_{21},M_{12},M_{43},Y_4^{i-1},X_4^i,Y_2^n,X_2^n,\notag \\
& \;\;\;\;\;\;\;\; Z_1^n,X_1^i)]-H(Z_4^n)\notag\\
&\overset{(b)}{\leq} \sum_{i=1}^n [H(Y_{2,i}|Y_2^{i-1},M_{21},X_{2,i})-H(Y_{2,i}|Y_2^{i-1},M_{12},M_{21},\notag\\
& \ \ M_{43},Z_1^n)+H(Y_{2,i}|Y_2^{i-1},M_{12},M_{21},M_{43},Z_1^n)-H(Z_{2,i})\notag \\
& \;\;\; +H(g_{34}X_{3,i}+Z_{4,i}|X_{4,i}, g_{32}X_{3,i}+Z_{2,i},X_1^i,X_2^n)-H(Z_{4,i})]\notag\\
&\overset{(c)}{\leq} \sum_{i=1}^nH(g_{12}X_{1,i}+g_{32}X_{3,i}+Z_{2,i}|X_{2,i})-H(Z_{2,i})\notag \\
&+H(g_{34}X_{3,i}+Z_{4,i}|X_{4,i}, g_{32}X_{3,i}+Z_{2,i})-H(Z_{4,i}) \label{(c)}
\end{align}}
In step (a), $X_1^i$ in the conditioning of the third term is constructed from ($M_{12},X_2^n,X_4^i,Z_1^n$). In step (b), we used conditioning reduces entropy, the second and the third term cancelled each other and $g_{32}X_{3,i}+Z_{2,i}$ in the conditioning of the fifth term is decoded from $Y_2^n$. In step (c), we only keep the self-interference $X_{4,i}$ and drop the terms  $X_1^i, X_2^n$ in the conditioning of the third term. We could leave these and express the outer bound in terms of correlation coefficients between the inputs (which in general may be correlated due to full adaptation). However, in subsequent steps we will seek to maximize, or outer bound this outer bound to obtain a simple analytical expression, which amounts to setting certain correlation coefficients to $0$, or equivalently, dropping the   terms  $X_1^i, X_2^n$ in the conditioning. 
Further evaluation yields \eqref{R_strong_E}, for details please refer to \cite[pg.41]{zcheng_arxiv}. 
\end{proof}

\begin{remark} 
{\it Sum-rate bound:} Note that the final, evaluated symmetric, normalized sum-rate bound in \eqref{R_strong_E} has the same form as the IC with perfect output feedback outer bound \cite[upper bound on (7)]{Changho2010}, though they are arrived at using slightly different genies. 

\end{remark}

\bigskip

\begin{theorem} {\it Outer bound: partial adaptation.}
For the two-way Gaussian IC under partial adaptation \eqref{eq:p1} -- \eqref{eq:p2}, in addition to the bounds in Theorem \ref{thm:outer-full}, any achievable rates ($R_{12},R_{21},R_{34},R_{43}$), and $R_{sym\rightarrow} = \frac{R_{12}+R_{34}}{2}$ and $R_{sym\leftarrow} = \frac{R_{21}+R_{43}}{2}$ must also satisfy, 
\begin{align}
&R_{12}\leq \log (1+{\tt SNR}_{12})\label{partial_single_rate1}\\
&R_{21}\leq \log (1+{\tt SNR}_{21})\\
&R_{34}\leq \log (1+{\tt SNR}_{34})\\
&R_{43}\leq \log (1+{\tt SNR}_{43})\label{partial_singe-rate4}
\end{align}
\begin{align}
R_{sym\rightarrow} 
 & \leq \log \left( 1+{\tt INR}+{\tt SNR}-\frac{{\tt INR}\times {\tt SNR}}{1+{\tt INR}}\right) \label{R_weak_E} 
 \end{align}
 \begin{equation}   R_{sym\leftarrow} \leq \left\{ \begin{array}{l}
\log \left(1+{\tt INR}+\frac{{\tt SNR}}{{\tt INR}}\right), \ \ \mbox{if} \ {\tt SNR}\leq {\tt INR}^3\\
\log \left(1+\frac{(\sqrt{{\tt SNR}}+\sqrt{{\tt INR}})^2}{1+{\tt INR}}\right), \ \ \mbox{if} \ {\tt SNR}>{\tt INR}^3
 \end{array}\right. \label{R_weak_BE}
\end{equation}
\label{thm:outer-partial}
\end{theorem}
\begin{proof}
For the single-rate bounds, it is sufficient to show the first two due to symmetry (notice that we must treat the $\rightarrow$ and $\leftarrow$ directions separately however due to the asymmetry of partial adaptation).
\begin{align*}
&n(R_{12}-\epsilon)\leq I(M_{12};Y_2^n|M_{21},M_{34})\\
&\leq H(Y_2^n|M_{21},M_{34})-H(Y_2^n|M_{21},M_{34},M_{12},X_1^n,X_2^n,X_3^n)\\
&\overset{(a)}{\leq} \sum_{i=1}^n [H(Y_{2,i}|Y_2^{i-1},M_{21},X_{2,i},M_{34},X_{3,i})-H(Z_{2,i})]\\
&\leq \sum_{i=1}^n [H(g_{12}X_{1,i}+Z_{2,i})-H(Z_{2,i})]\\
&\leq \sum_{i=1}^n [\log (1+{\tt SNR}_{12})] 
\end{align*}
\begin{align*}
 &n(R_{21}-\epsilon)\leq I(M_{21};Y_1^n|M_{12},M_{43},M_{34},Z_4^{n-1})\\
& \leq H(Y_1^n|M_{12},M_{34},M_{43},Z_4^{n-1})\\
&\;\;\;\;- H(Y_1^n|M_{12},M_{34},M_{43},Z_4^{n-1},M_{21},X_1^n,X_2^n,X_4^n)\\
&\overset{(b)}{\leq} \sum_{i=1}^n [H(Y_{1,i}|M_{12},M_{34},M_{43},Z_4^{n-1},Y_1^{i-1},X_{1,i},X_{4,i})\\
& \;\;\;\;\;\;-H(Z_{1,i})]\\
&\leq \sum_{i=1}^n [H(g_{21}X_{2,i}+Z_{1,i})-H(Z_{1,i})]\\
&\leq \sum_{i=1}^n [\log (1+{\tt SNR}_{21})]
\end{align*}
where (a) follows from the definition of partial adaptation and (b) follows similarly, and by Lemma \ref{lemma:partialG}.

For the $\rightarrow$ direction of the symmetric rate, due to space constraints, we refer the reader to the in-detail converse on pg. 28 of \cite{zcheng_arxiv}; the key starting and ending steps are as follows:
\begin{align}
&n(R_{12}+R_{34}-\epsilon)\leq I(M_{12};Y_2^n,g_{14}X_1^n+Z_4^n,M_{21},M_{43})\notag \\
& \;\;\;\; +I(M_{34};Y_4^n,g_{32}X_3^n+Z_2^n,M_{21},M_{43})\notag \\
& \leq \sum_{i=1}^n [H(g_{12}X_{1,i}+g_{32}X_{3,i}+Z_{2,i}|g_{14}X_{1,i}+Z_{4,i},X_{2,i})\notag \\
& \ \ +H(g_{34}X_{3,i}+g_{14}X_{1,i}+Z_{4,i}|g_{32}X_{3,i}+Z_{2,i},X_{4,i})\notag \\
& \;\;\;\;\;\;-H(Z_{2,i})-H(Z_{4,i})] \label{last}
\end{align}
In the first step, we have given $(g_{14}X_1^n+Z_4^n)$ and $(g_{32}X_3^n+Z_2^n)$ as side information; in the intermediate steps (\cite[pg. 28]{zcheng_allerton}), we have used the definition of partial adaptation and cancellation of certain negative entropy terms. 


To obtain \eqref{R_weak_E} we continue to outer bound \eqref{last} in terms of ${\tt SNR}$ and ${\tt INR}$, using the fact that Gaussians maximize entropy subject to variance constraints. Specifically, one may intuitively see that, if one defines $\lambda_{jk} = E[X_j X_k^*]$, that one may express \eqref{last} in terms of $\lambda_{12}, \lambda_{13}, \lambda_{14}, \lambda_{34}, \lambda_{23}$. One also notices from the conditional entropy expression in \eqref{last} that taking $ \lambda_{14} = \lambda_{23} = \lambda_{12}  = \lambda_{34} = 0$, and since $\lambda_{13} = 0$ (naturally, by partial adaptation) will maximize the outer bound. This may alternatively be worked out by calculating the conditional covariance matrices directly (as we will show for the next bound on $R_{\leftarrow}$). In this case then,  for each $i$, we may bound
\begin{align*}
&H(g_{12}X_{1}+g_{32}X_{3}+Z_{2}|g_{14}X_{1}+Z_{4},X_{2}) - H(Z_2) \\
& \leq H(g_{12}X_{1}+g_{32}X_{3}+Z_{2}|g_{14}X_{1}+Z_{4}) - H(Z_2) \\
&\leq \log 2\pi e(\mbox{Var}(g_{12}X_{1}+g_{32}X_{3}+Z_{2}|g_{14}X_{1}+Z_{4}))\\
& \;\;\;\; -\log 2\pi e(\mbox{Var}(Z_2)) \\
& \leq \log\left(1+{\tt SNR}+{\tt INR} - \frac{{\tt SNR}\times {\tt INR}}{1+{\tt INR}}\right),
\end{align*}
which together with the symmetric expressions for the second and fourth terms in \eqref{last} yield \eqref{R_weak_E}.

\bigskip
For the $\leftarrow$ direction, we again defer the reader to \cite[pg.29]{zcheng_allerton} due to space constraints, but we are 
similarly able to obtain:
\begin{align}
&n(R_{21}+R_{43}-\epsilon)\leq I(M_{21};Y_1^n,g_{23}X_2^n+Z_3^n,M_{12},M_{34})\notag \\
&+I(M_{43};Y_3^n,g_{41}X_4^n+Z_1^n,M_{12},M_{34})\notag \\
&\leq \sum_{i=1}^n [H(g_{21}X_{2,i}+g_{41}X_{4,i}+Z_{1,i}|g_{23}X_{2,i}+Z_{3,i},X_{1,i})\notag \\
& \ \ +H(g_{43}X_{4,i}+g_{23}X_{2,i}+Z_{3,i}|g_{41}X_{4,i}+Z_{1,i},X_{3,i})\notag \\
& \;\;\;\; -H(Z_{1,i})-H(Z_{3,i})]\label{last2}
\end{align}
There are some slight differences in the converse,  compared to the previous outer bound due to the partial adaptation constraints (and hence more care must be taken when constructing $X_{2,i}, X_{4,i}$). 

We again proceed to outer bound \eqref{last2} to obtain \eqref{R_weak_BE}. 
It is sufficient to evaluate the first and third terms in \eqref{last2} due to symmetry. We could outer bound \eqref{last2} in terms of the conditional covariance matrices and then proceed to select values of the correlation coefficients (complex) $\lambda_{jk} : = E[X_j X_k^*]$  which maximize this outer bound. A more intuitive method is to note that again, the conditional entropies in \eqref{last2} will be maximized if $\lambda_{14} = \lambda_{32}=0$, and $\lambda_{12}=\lambda_{34}=0$, which may also be obtained by dropping $X_{1,i}, X_{3,i}$ in the conditioning terms. At that point, we are only left with the coefficient $\lambda_{24} = E[X_2X_4^*]$, (which in contrast to the $\rightarrow$ bound is not automatically $0$ due to the possible adaptation in the $\leftarrow$ direction. Furthermore, setting it to zero cannot be argued intuitively as we see a tradeoff.) yielding the following bound for $R_{sym \leftarrow} = \frac{R_{21}+R_{43}}{2}$ by symmetry: 
\begin{align}
&R_{sym\leftarrow}
\leq H(g_{21}X_{2}+g_{41}X_{4}+Z_{1}|g_{23}X_{2}+Z_{3})-H(Z_{1})\notag \\
&\leq \log 2\pi e\left(\mbox{Var}(g_{21}X_{2}+g_{41}X_{4}+Z_{1}|g_{23}X_{2}+Z_{3})\right)\notag \\
& \;\;\;\; -\log 2\pi e (\mbox{Var}(Z_1))\notag \\
&\leq \log \left( 1+{\tt INR}+{\tt SNR}+2|\lambda_{24}|\cos \theta\sqrt{{\tt SNR}\times {\tt INR}}\right. \notag\\
&\left. -\frac{{\tt SNR}\times {\tt INR} + {\tt INR}^2|\lambda_{24}|^2 + 2\sqrt{{\tt SNR}}{\tt INR}^{3/2}|\lambda_{24}|\cos\theta}{1+{\tt INR}}\right)\label{lambda}
\end{align}
where $\theta$ is the angle of $g_{21}g_{41}^*\lambda_{24}$. To maximize \eqref{lambda}, we take the partials of the expression with respect to $|\lambda_{24}|$ and $\theta$ and set these to 0. For these to equal 0 for all ${\tt SNR}$ and ${\tt INR}$ we must have $\theta = 0$ and  $|\lambda_{24}|=\frac{\sqrt{{\tt SNR}\times {\tt INR}}}{{\tt INR}^2}$ (discussed next).  
Note that we must constrain $|\lambda_{24}|\in [0,1]$. In the interval $|\lambda_{24}| \in \left[0, \frac{\sqrt{\tt SNR\times INR}}{{\tt INR}^2}\right]$ one may verify that the function is increasing in $|\lambda_{24}|$. Thus, if $ \frac{\sqrt{\tt SNR\times INR}}{{\tt INR}^2} \leq 1$, $(|\lambda_{24}| =  \frac{\sqrt{\tt SNR\times INR}}{{\tt INR}^2}, \theta =0)$ maximizes \eqref{lambda}; this happens if ${\tt SNR}\leq {\tt INR}^3$, and yields the first bound in \eqref{R_weak_BE}. Otherwise, for ${\tt SNR}>{\tt INR}^3$, $(\lambda_{24} = 1, \theta=0)$ maximizes \eqref{lambda}, yielding the second equation in \eqref{R_weak_BE}. 
\end{proof}

\begin{remark} 
 The sum-rate bound for $R_{sym\rightarrow}$ of \eqref{R_weak_E} has the same form as Etkin, Tse and Wang's outer bound for one-way Gaussian interference channel \cite[(12)]{etkin_tse_wang} which is useful in weak interference.  The sum-rate bound for $R_{sym\leftarrow}$ is quite different, and we note that it may be verified that \eqref{R_weak_BE} is always at least as large as \eqref{R_weak_E}, as one might expect given the partial adaptation constraints on nodes in the $\rightarrow$ direction, but none on the nodes in the $\leftarrow$ direction.   
\end{remark}

\section{Capacity to within a constant gap}

We now demonstrate that these outer bounds, derived for the fully adaptive or partially adaptive models, may be achieved to within a constant gap or capacity by {\it non-adaptive} schemes -- i.e. simultaneous decoding or the Han and Kobayashi scheme operating in the two directions independently. We break our analysis into three sub-sections: 1) very strong interference, 2) strong interference, and 3) weak interference. The overall finite gap results are summarized in Table \ref{table:gaps}.

\subsection{Very Strong Interference: ${\tt INR}\geq {\tt SNR}(1+{\tt SNR})$}
We first show that a non-adaptive scheme may achieve the capacity for the two-way Gaussian IC under a partially adaptive model in very strong interference. For the symmetric two-way Gaussian IC, define ``very strong interference'' as the class of channels for which ${\tt INR}\geq {\tt SNR}(1+{\tt SNR})$, as in \cite[below equation (21)]{etkin_tse_wang}.
%
It is well known that the capacity region of the one-way Gaussian IC in very strong interference is that of two parallel Gaussian point-to-point channels \cite{Carleial1975}, which may be achieved by having each receiver first decode the interfering signal, treating its own as noise, subtracting off the decoded interference, and decoding its own message. Given that the interference is so strong,  this may be done without a rate penalty. We ask whether the same is true for the two-way Gaussian IC with partial adaptation. The answer is affirmative and the capacity region is given by the following theorem:


%

%




\begin{theorem}
\label{2wIC_very_strong}
The capacity region for the two-way Gaussian interference channel with partial adaptation in very strong interference is the set of rate pairs ($R_{12},R_{21},R_{34},R_{43}$), such that \eqref{partial_single_rate1}--\eqref{partial_singe-rate4} are satisfied. 
\end{theorem}
\begin{proof}
Each node may ignore its ability to adapt,  and rather transmit using a ${\cal CN}(0,1)$ Gaussian random code. Each receiver may cancel its own self-interference, and then proceed to decode first the single interfering term before decoding its own message. This standard non-adaptive scheme may achieve the outer bound in \eqref{partial_single_rate1}--\eqref{partial_singe-rate4} in Theorem \ref{thm:outer-partial}.
\end{proof}

Interestingly, the capacity region of the two-way Gaussian interference channel with partial adaptation in very strong interference, is equivalent to the capacity regions of two one-way Gaussian interference channels with very strong interference in parallel and is achieved using a non-adaptive scheme. This allows us to conclude that {\it partial} adaptation is useless in this symmetric and very strong interference regime.

\subsection{Strong Interference: ${\tt SNR}\leq {\tt INR}\leq {\tt SNR}(1+{\tt SNR})$}
In this regime, we are able to show that a non-adaptive scheme may achieve capacity to within a constant gap of any {\it fully} adaptive scheme (in contrast to any {\it partially} adaptive scheme in the last subsection).  A symmetric two-way Gaussian IC,  as in \cite{etkin_tse_wang}, is said to be in ``strong interference'' when ${\tt INR}\geq {\tt SNR}$.

The capacity region of one-way Gaussian interference channel in strong interference is given by \cite{sato_strong}, and for symmetric channels, the capacity region when the interference is strong but not very strong, i.e. ${\tt SNR}\leq {\tt INR}\leq {\tt SNR}(1+{\tt SNR})$,  may be written as 
\begin{align}
R_{sym} = \frac{R_{12}+R_{34}}{2} \leq \frac{1}{2}\log (1+{\tt SNR}+{\tt INR}).\label{R_sato}
\end{align} 
We note that this rate is achievable for the two-way Gaussian IC  by using the simultaneous non-unique decoding scheme for the interference channel in strong interference \cite{ElGamalKim:book, sato_strong, ahlswede1974}) in the $\rightarrow$ and $\leftarrow$ directions, and noting that any self-interference may be canceled. This is a non-adaptive scheme.

We will show that this non-adaptive scheme which achieves \eqref{R_sato} in each direction (i.e. $R_{sym} =$\eqref{R_sato}) also achieves to within 1 bit (per user, per direction) of our fully adaptive outer bound \eqref{R_strong_E} in strong but not very strong interference. 

\begin{theorem}
The capacity region for two-way symmetric Gaussian interference channel with full adaptation in strong (but not very strong) interference is within 1 bit to \eqref{R_sato} (per user, per direction).
\end{theorem}
\begin{proof}
\begin{align*}
& \eqref{R_strong_E} - \eqref{R_sato}\\
&\overset{(a)}{\leq} \frac{1}{2} \log 2(1+{\tt SNR}+{\tt INR})+\frac{1}{2}\log \left(1+\frac{{\tt SNR}}{1+{\tt INR}}\right)\\
&-\frac{1}{2}\log (1+{\tt SNR}+{\tt INR})\\
&\overset{(b)}{\leq} \frac{1}{2}+\frac{1}{2}\log \left(1+\frac{{\tt INR}}{{\tt INR}}\right)=1
\end{align*}
In step (a), we use the fact that $1+{\tt SNR}+{\tt INR}+2\sqrt{{\tt SNR}\times {\tt INR}}\leq 2(1+{\tt SNR}+{\tt INR})$. Step (b) follows from the condition of strong interference ${\tt INR}\geq {\tt SNR}$. Notice that the bound \eqref{R_strong_E} is valid for the symmetric assumptions of full adaptation;  we thus conclude that the non-adaptive schemes' gap to the fully adaptive outer bound for each user, for each direction is at most 1 bit.
\end{proof}

\subsection{Weak Interfererence: ${\tt INR}\leq {\tt SNR}$}
In the following we demonstrate that the well known Han and Kobayashi scheme employed in parallel in the $\rightarrow$ and $\leftarrow$ directions may achieve to within a constant number of bits of the fully or partially adaptive (depends on the channel regimes, or relative ${\tt SNR}$ and ${\tt INR}$ values) capacity region for the two-way Gaussian IC. 

\begin{theorem}
 A non-adaptive scheme may achieve to within a $2$ bit per user per direction of partially adaptive capacity region for the two-way Gaussian IC in weak interference. In some channel regimes, this non-adaptive scheme also achieves to within a constant gap of any {\it fully} adaptive scheme.
\end{theorem}
\begin{proof}
As for the one-way IC \cite{etkin_tse_wang}, we break our proof into two regimes: ${\tt INR}\geq 1$ or ${\tt INR}<1$. 

\subsubsection{${\tt INR}\geq 1$}
Outer bounds have already been derived. Consider now using the specific choice of the Han and Kobayashi  (HK) strategy utilized for the symmetric one-way IC as in \cite[(4)]{etkin_tse_wang}  in each direction. That is, view nodes 1,2 as transmitters and 3,4 as receivers in the $\rightarrow$ direction and employ the particular choice of the HK scheme where private messages are encoded at the level of the noise, and similarly for the $\leftarrow$ direction consider nodes 3,4 as transmitters and 1,2 as receivers. Due to the additive nature of the channel and each node's ability to first cancel out their self-interference, one may achieve the following rates per user, per node for each direction when ${\tt INR}\geq 1$ for the symmetric two-way Gaussian IC:
{\small \begin{align}
& R_{HK}=\min \left\{\frac{1}{2}\log (1+{\tt INR}+{\tt SNR})+\frac{1}{2}\log \left(2+\frac{{\tt SNR}}{{\tt INR}}\right)-1, \right.\notag \\
&\left. \log \left(1+{\tt INR}+\frac{{\tt SNR}}{{\tt INR}}\right)-1\right\} =: \min\{R_{HK1}, R_{HK2}\}. \label{HK}
\end{align}}

\smallskip
\noindent {\bf If the first term in \eqref{HK} is active} we show a constant gap to the outer bound \eqref{R_strong_E},
\begin{align*}
& \eqref{R_strong_E}-R_{HK1}\\
&\leq \frac{1}{2} \log 2(1+{\tt SNR}+{\tt INR})-\frac{1}{2}\log (1+{\tt INR}+{\tt SNR})\\
&+\frac{1}{2}\log \left(1+\frac{{\tt SNR}}{{\tt INR}}\right)-\frac{1}{2}\log \left(2+\frac{{\tt SNR}}{{\tt INR}}\right)+1\\
&\leq \frac{1}{2}\log (2)+\frac{1}{2}\log (1)+1=1.5
\end{align*}

Since our bound \eqref{R_strong_E} is derived assuming full adaptation,  we may conclude that this gap holds for both $R_{sym\rightarrow}$ and $R_{sym\leftarrow}$ (i.e. holds for $R_{sym}$). 

 \smallskip
\noindent {\bf If the second term in \eqref{HK} is active,} {we use outer bound \eqref{R_weak_E} for the forward direction, to bound the gap for $R_{sym\rightarrow}$ as }
\begin{align*}
& \eqref{R_weak_E}-R_{HK2} \\
&= \log \left(\frac{{\tt INR}(1+{\tt INR})^2+{\tt SNR}\times {\tt INR}}{{\tt INR}(1+{\tt INR})^2+{\tt SNR}(1+{\tt INR})}\right)+1\\
&\leq \log (1)+1=1
\end{align*}
Since our bound \eqref{R_weak_E}  has the same form as the ETW bound \cite{etkin_tse_wang}, the capacity of the two-way Gaussian interference channel with partial adaptation in the forward direction is also to within 1 bit of the specific HK rate \eqref{HK} when ${\tt INR}\geq 1$.


{We use outer bound \eqref{R_weak_BE} for the backward direction, to bound the gap for $R_{sym\leftarrow}$,} noting that we need to consider both cases separately. 
{If the first term in \eqref{R_weak_BE} is relevant (${\tt SNR}\leq {\tt INR}^3$), one may easily conclude that $ \eqref{R_weak_BE} -R_{HK2} = 1$.}
{If the second term in \eqref{R_weak_BE} is relevant (${\tt SNR}\geq {\tt INR}^3$):} 
\begin{align*}
&  \eqref{R_weak_BE} -R_{HK2}\\
&\overset{(a)}{\leq} \log\left(\frac{2({\tt INR}+{\tt SNR}\times {\tt INR}+2{\tt INR}^2+{\tt SNR}+{\tt INR}^3)}{{\tt INR}+{\tt SNR}\times {\tt INR}+2{\tt INR}^2+{\tt SNR}+{\tt INR}^3}\right)+1\\
& = \log (2)+1=2
\end{align*} 
where (a) follows the fact that $1+{\tt SNR}+{\tt INR}+2\sqrt{{\tt SNR}\times {\tt INR}}\leq 2(1+{\tt SNR}+{\tt INR})$, and additional details may be found in \cite[pg.34]{zcheng_arxiv}.


\subsubsection{${\tt INR}<1$} In this case, a symmetric version of the HK scheme may be obtained 
from  \cite[(69)]{etkin_tse_wang}, for which each of the four users may achieve the following rate:
\begin{align}
& R_{{\tt INR}<1}\leq \log \left(1+\frac{{\tt SNR}}{1+{\tt INR}}\right)\label{R_INR<1}
\end{align}
{This achieves to within 1 bit of the outer bound \eqref{R_strong_E}:} 
\begin{align*}
&\eqref{R_strong_E}-R_{{\tt INR}<1}\\
& \leq \frac{1}{2}\log \left(\frac{2(1+{\tt SNR}+{\tt INR})(1+{\tt INR})}{1+{\tt SNR}+{\tt INR}}\right)\\
& \overset{(a)}{\leq} \frac{1}{2}\log (4)=1
\end{align*}
where (a) we use the condition of ${\tt INR}<1$. Since \eqref{R_strong_E} was obtained for full adaptation, we can conclude that the capacity of the two-way Gaussian IC is to within 1 bit to the HK region when ${\tt INR}<1$ for both directions.
\end{proof}

We summarize the constant gaps in Table \ref{table:gaps}.
\begin{table}
\resizebox{9cm}{!} {
\begin{tabular}{|c|c|c|c|c|c|}
\hline
{\bf Interference} & \multicolumn{5}{r|}{{\bf Constant Gaps per user per direction (bits)}}\\ \hline
Very Strong & \multicolumn{5}{r|}{0 (partial)}\\ \hline
Strong & \multicolumn{5}{r|}{1 (full)}\\ \hline
 & ${\tt INR}<1$ & \multicolumn{4}{r|}{1 (full)}\\ \cline{2-6}
 &  & HK1  active & \multicolumn{3}{r|}{1.5 (full)}\\ \cline{3-6}
 Weak & ${\tt INR}\geq 1$ &  & $\rightarrow$& \multicolumn{2}{r|}{1 (partial)}\\ \cline{4-6}
  & & HK2 active & $\leftarrow$ & ${\tt SNR}\leq {\tt INR}^3$ & 1 (partial)\\ \cline{5-6}
  & & & &  ${\tt SNR}> {\tt INR}^3$ & 2 (partial)\\ \hline
\end{tabular}
}
\caption{Constant gaps between non-adaptive symmetric Han and Kobayashi schemes in each direction and partially or fully adaptive outer bounds.}
\label{table:gaps}
\end{table}


\section{Conclusion}
\label{conclusion}
We have introduced the two-way Gaussian interference channel; obtained outer bounds under full and partial adaptation constraints, and shown that simple non-adaptive schemes (including the Han and Kobayashi scheme) achieve to within a constant gap for the symmetric sum-rate of these fully or partially adaptive outer bounds. We do not believe that in general, non-adaptive schemes will achieve to within a constant gap of capacity for general non-symmetric and fully adaptive two-way Gaussian ICs, this is an interesting question which is the topic of ongoing work.

\bibliographystyle{IEEEtran}
\bibliography{refs}

\end{document}